\documentclass[preprint,12pt,numbers]{elsarticle}

\usepackage{graphicx}
\usepackage{amssymb}
\usepackage{amsmath}
\usepackage{amsthm}
\usepackage{bm}
\usepackage{xcolor}
\usepackage{multirow}
\usepackage{url}
\usepackage{algorithm}
\usepackage{algpseudocode}

\makeatletter
\newcommand{\figcaption}[1]{\def\@captype{figure}\caption{#1}}
\newcommand{\tabcaption}[1]{\def\@captype{table}\caption{#1}}
\newcommand\footnoteref[1]{\protected@xdef\@thefnmark{\ref{#1}}\@footnotemark}
\makeatother
\newcommand{\R}{\mathbb{R}}

\newcommand{\bbR}{\mathbb{R}}

\newcommand{\norm}[1]{{\left\| {#1} \right\|}}

\theoremstyle{definition}

\newtheorem{proposition}{Proposition}

\newtheorem{lemma}{Lemma}
\newtheorem{rem}{Remark}

\newcommand{\bbE}{\mathbb{E}}

\newcommand{\calN}{\mathcal{N}}
\newcommand{\calU}{\mathcal{U}}

\newcommand{\st}{\mathrm{subject\ to\ }}
\newcommand{\NN}{\mathrm{NN}}
\newcommand{\res}{\mathrm{res}}

\newcommand{\clip}{\mathrm{clip}}
\newcommand{\rmP}{\mathrm{P}}
\newcommand{\rmI}{\mathrm{I}}
\newcommand{\rmD}{\mathrm{D}}
\newcommand{\sat}{\mathcal{S}}

\algrenewcommand\algorithmiccomment[1]{\hfill{\footnotesize // #1}}

\begin{document}

% Elsevier template wrapper (kept empty so that the manuscript body is not altered)
\begin{frontmatter}

\title{Model Predictive Path Integral PID Control for Learning-Based Path Following}

\author[aff1]{Teruki Kato\corref{cor1}}
\ead{teruki.kato.tg@mosk.tytlabs.co.jp}

\author[aff1]{Koshi Oishi}
\author[aff1]{Seigo Ito}

\cortext[cor1]{Corresponding author.
Postal address: Toyota Central R\&D Labs., Inc., Nagakute 480-1192, Japan.
Phone: +81-90-6384-4818}

\affiliation[aff1]{organization={Toyota Central R\&D Labs., Inc.},
            addressline={Nagakute 480-1192},
            country={Japan}}
            
\begin{abstract}
Classical proportional--integral--derivative (PID) control remains widely used in industrial control systems, while model predictive control (MPC) is actively studied to achieve higher performance for systems with nonlinear dynamics.
Model predictive path integral (MPPI) control is a sampling-based MPC method that optimizes control inputs without gradient calculations and can handle non-differentiable models and objective functions.
However, conventional MPPI directly samples control-input sequences, which can produce large temporal input increments and causes the optimization dimension to grow with the prediction horizon.
This study proposes MPPI--PID control, which uses MPPI to optimize PID gains online instead of directly optimizing the control-input sequences.
By replacing high-dimensional input-sequence optimization with low-dimensional gain-space optimization while retaining the PID structure, the proposed formulation improves sampling efficiency and promotes smoother control inputs.
Theoretical analyses are provided for a unified path-integral update, the relation between optimization dimension and effective sample size, and the temporal correlation of input perturbations induced by the PID structure.
The method is evaluated on a learning-based path following of a mini forklift using a residual-learning dynamics model that combines a physical model and a neural network identified from real-machine driving data.
Numerical results show that MPPI--PID improves tracking performance over fixed-gain PID, yields smaller input increments than conventional MPPI, and maintains favorable performance under reduced sampling budgets.
\end{abstract}


\end{frontmatter}

\section{Introduction}
Classical proportional--integral--derivative (PID) control remains a central and high-impact technology in industrial control, while model predictive control (MPC) is regarded as a highly promising advanced-control technology for future industrial applications \cite{samad2020industry}.
This practical situation motivates control methods that combine the deployability and interpretability of a PID control with the predictive optimization capability of MPC.
Particularly, the online tuning of PID gains inside a receding-horizon framework is attractive because it can improve performance without discarding the controller architecture, which is already familiar in industrial practice.

MPC computes control inputs by solving a finite-horizon optimal control problem at each control step \cite{mpc_handbook}.
Gradient-based optimization is widely used for this purpose, but applying it when the plant model or objective function is non-differentiable is difficult.
Model predictive path integral (MPPI) control \cite{mppi_information} is a representative sampling-based MPC method that avoids gradient calculations and therefore can naturally handle neural-network models \cite{brunton_kutz} and objectives with logical branching \cite{mppi_stl}.
Conversely, conventional MPPI directly samples the control-input sequence.
This direct-sequence optimization can lead to large temporal input increments and causes the optimization dimension to grow in proportion to the prediction horizon.

Several studies have optimized PID gains rather than the control inputs within receding-horizon schemes \cite{pid_mpc_1,pid_mpc_2,pid_mpc_3,pinn_pid}.
These methods retain the PID structure and reduce the number of optimization variables, but their gradient-based formulations are less suitable for non-differentiable models and objective functions.
Reinforcement learning (RL) methods, such as path integral policy improvement (PI$^2$) \cite{pi2}, have also been applied to PID-related policy-parameter learning \cite{pi2_pid}, but they are typically model-free and episode-based rather than model-based and receding horizon.
Table~\ref{tab:method_comparison} summarizes the qualitative position of the proposed method relative to related approaches.

\begin{table}[t]
  \centering
  \small
  \tabcaption{Qualitative comparison of related control methods}
  \label{tab:method_comparison}
  \setlength{\tabcolsep}{5pt}
  \begin{tabular}{@{}p{0.28\linewidth}p{0.14\linewidth}p{0.09\linewidth}p{0.14\linewidth}p{0.24\linewidth}@{}}
    \hline
    Method & Framework & PID law & Nonsmooth\newline objectives & Optimized-variable\newline dimension \\
    \hline
    Standard MPC \cite{mpc_handbook} & MPC & No & No & Horizon-dependent \\
    Standard MPPI \cite{mppi_information} & MPC & No & Yes & Horizon-dependent \\
    MPC--PID \cite{pid_mpc_1,pid_mpc_2,pid_mpc_3,pinn_pid} & MPC & Yes & No & Horizon-independent \\
    MPPI--PID (proposed) & MPC & Yes & Yes & Horizon-independent \\
    PI$^2$--PID \cite{pi2_pid} & RL & Yes & Yes & Episode-length-\newline independent \\
    \hline
  \end{tabular}
\end{table}

This study proposes the MPPI--PID control, which optimizes PID gains at each control step using MPPI.
The gain vector, rather than the control-input sequence, is sampled and updated based on the predicted finite-horizon costs.
Then, the resulting control input is generated by a multivariable PID law.
Thus, the method preserves the PID-controller structure while inheriting the gradient-free optimization capability of MPPI.
Figure~\ref{fig:proposed_method} illustrates the overall idea when the controller is combined with a residual-learning dynamics model.

\begin{figure}[!htbp]
  \centering
  \includegraphics[width=0.95\linewidth]{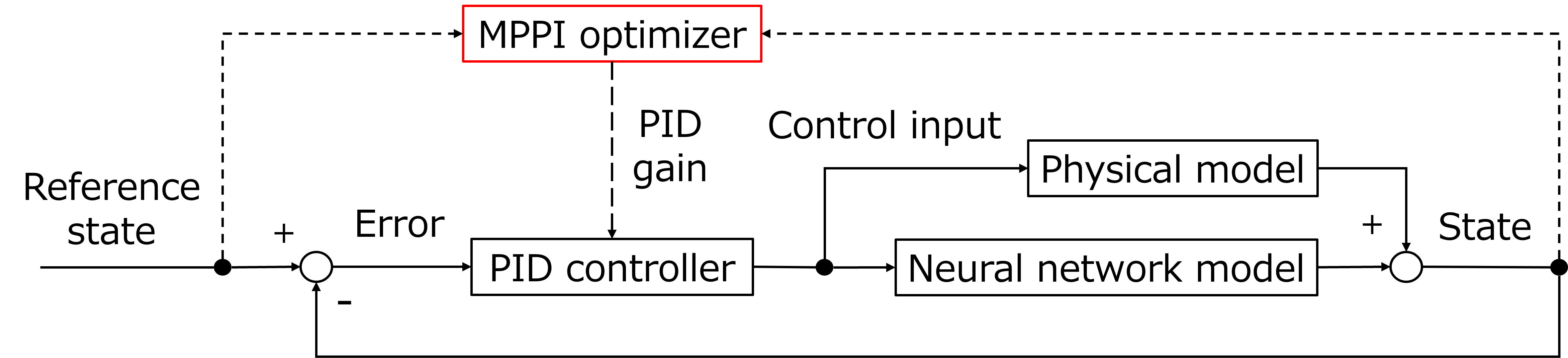}
  \caption{Overview of MPPI--PID with a residual-learning dynamics model}
  \label{fig:proposed_method}
\end{figure}

A summary of the contributions of this paper is the following:
\begin{itemize}
    \item We formulate MPPI--PID, a model-based receding-horizon controller that optimizes PID gains by sampling while retaining the PID control law.
    \item We express conventional MPPI and MPPI--PID through a common abstract-optimization variable. For conventional MPPI, this viewpoint recovers the information-theoretic path-integral update \cite{mppi_information}; for MPPI--PID, the same update is applied to the PID-gain vector.
    \item Building on this unified view, we analyze how the optimization dimension affects effective sample size and how the PID structure induces temporal correlation of control-input perturbations, thereby supporting the expected improvements in sampling efficiency and input smoothness.
    \item We evaluate the method on a learning-based path following of a mini forklift using a residual-learning dynamics model identified from real-machine driving data and compare its performance with the fixed-gain PID control and conventional MPPI.
\end{itemize}

The symbols used in this study are as follows.
The element-wise product of vectors is denoted by $\odot$, and element-wise inequalities are expressed using $\leq$.
The function $\varphi^{\clip}_{[l,u]}:\R^n\to\R^n\ (l, u\in\R^n)$ denotes the function that clips each component $z_j$ of a vector $z\in\R^n$ to the interval $[l_j,u_j]$.
$\calN(\mu,\Sigma)$ denotes a normal distribution with the mean vector $\mu$ and covariance matrix $\Sigma$.
%

% The structure of the remaining paper is as follows. The next section describes the methods implemented in this study. This is followed by the Results, which consists of all results obtained. Then, the Discussion includes an analysis of the results. Lastly, the Conclusion consists of the inferences drawn.
%
\section{Methods}
\subsection{Modeling}
The modeling objective in this study was to exploit known physical structures while learning residual dynamics that are difficult to represent analytically.
We used the following discrete-time residual-learning model as the prediction model for MPC:
\begin{align}
    x_{t+1}=f^\mathrm{phys}( x_t,u_t ) + w^\res(x_t) M^\res \odot f^{\NN,\res}( x_t,u_t )
    =: f( x_t,u_t ).
    \label{eq:residual_model}
\end{align}
Here, $x_t\in\R^{n_x}$ and $u_t\in\R^{n_u}$ denote the state and control input at time step $t$, respectively.
The physical model $f^\mathrm{phys}$ provides the nominal state transition, the neural network $f^{\NN,\res}$ learns the residual, $M^\res\in\{0,1\}^{n_x}$ specifies the state dimensions to which the residual is applied, and $w^\res:\R^{n_x}\to[0,1]$ adjusts the residual contribution depending on the state.
This structure follows the concept of learning corrections to a known model, as in universal differential equations \cite{universal_de}, while also allowing the residual to be restricted to selected state components \cite{pavone_scaramuzza_pinn,tri_pinn,bosch_pinn} or suppressed outside reliable data regions \cite{pinn_domain_of_validity}.

The physical model is first described in continuous time as follows:
\begin{align}
    \dot{x}(t)=f^\mathrm{phys,c}( x(t),u(t) ).
\end{align}
Applying Euler discretization with step size $h$ gives the following:
\begin{align}
    x_{t+1}=x_t+f^\mathrm{phys,c}( x_t,u_t )h =: f^\mathrm{phys}( x_t,u_t ).
\end{align}
For comparison with the residual-learning model, a purely neural discrete-time model can also be written as the following:
\begin{align}
    x_{t+1}=f^\NN( x_t,u_t ).
\end{align}
When some state components are already represented accurately by $f^\mathrm{phys}$, the mask $M^\res$ allows the neural network to correct only the remaining components.
The weighting function $w^\res$ enables the model to prioritize the physical formulation in data-scarce regions while using the learned residual in data-rich regions.

\subsection{Control}
\subsubsection{PID control}
A PID control computes the control input based on proportional, integral, and derivative operations applied to the deviation of the plant state from a reference state.
In this study, we used a multivariable PID controller in which each control-input component could depend on multiple error signals.
The control law is the following:
\begin{align}
    u_{t,j} = u_j^\mathrm{bias} + \sum_{l=1}^{n_j^e}\left(
     K_{j,l}^\rmP e_{t,j,l} + K_{j,l}^\rmI \sum_{\tau=0}^{t} e_{\tau,j,l} h + K_{j,l}^\rmD \frac{e_{t,j,l} - e_{t-1,j,l}}{h}
    \right).
    \label{eq:pid}
\end{align}
Here, $u_{t,j}$ denotes the $j$-th component of the control input at time $t$, with $u_j^\mathrm{bias}$ representing a bias term.
$e_{t,j,l}$ denotes the $l$-th error used for the control input $u_{t,j}$, while $K_{j,l}^\rmP$, $K_{j,l}^\rmI$ and $K_{j,l}^\rmD$ are the PID gains.
When upper- and lower-bound constraints existed on both the control input and its rate of change, $\underline{u} \leq u_t \leq \overline{u}$ and $\underline{\Delta u} \leq u_t - u_{t-1} \leq \overline{\Delta u}$, we applied the following saturation mapping:
\begin{align}
    \sat(u_t;u_{t-1}) :=
    \varphi^{\clip}_{[\underline u,\overline u]}
    \left(
        u_{t-1}+\varphi^\clip_{ [ \underline{\Delta u}, \overline{\Delta u} ] }
        \left( u_t - u_{t-1} \right)
    \right).
    \label{eq:saturation_map}
\end{align}
This mapping initially saturated the one-step change and then saturated the input itself.
When a predicted control sequence was generated, the mapping was applied sequentially with respect to the prediction time because the admissible set at time $t$ depended on the preceding input $u_{t-1}$.
Defining the constraint set as $\calU(u_{t-1}) := \{ u\in\R^{n_u} | \underline{u} \leq u \leq \overline{u},\ \underline{\Delta u} \leq u - u_{t-1} \leq \overline{\Delta u} \}$, the saturation mapping coincided with the projection $\Pi_{ \calU(u_{t-1}) }(u_t)$, as shown in the Appendix.

\subsubsection{MPPI}
In MPC, denoting the current state at each control step as $x_0$, we solved the optimal control problem with prediction horizon $N$ as follows:
\begin{align}
    \min_{u}\ J := \sum_{t=0}^{N-1} l_t
    \ \st\
    x_{t+1}=f(x_t,u_t),\
    u_t \in \calU(u_{t-1})\ (t=0,\ldots,N-1)
    \label{eq:mppi_ocp}
\end{align}
and applied the obtained $u_0$ to the plant.
Here, $l_t$ is the stage cost at time $t$.
When the state-transition function $f$ or objective function $J$ in the preceding problem was non-differentiable, conventional gradient-based optimization methods were not directly applicable.
Therefore, MPPI performed approximate optimization by sampling the control-input sequence.
Let $\lambda>0$ be the temperature parameter. For samples with objective values $J^{(i)}$, the MPPI weight was computed as follows:
\begin{align}
    w^{(i)}=\exp\left(-\frac{J^{(i)}-\min_i J^{(i)}}{\lambda}\right)\ (i=1,\ldots,I).
    \label{eq:mppi_weights}
\end{align}
The conventional MPPI algorithm used in this study is summarized in Algorithm~\ref{alg:mppi}.

\begin{algorithm}[t]
\caption{Conventional MPPI for control-input sequence optimization}
\label{alg:mppi}
\begin{algorithmic}[1]
\Require Current state $x_0$, previous input $u_{-1}$, nominal input sequence $u_0,\ldots,u_{N-1}$
\For{$k=1,\ldots,N^\mathrm{iter}$}
    \For{$i=1,\ldots,I$}
        \State Sample $\varepsilon_t^{(i)}\sim\calN(0,\mathrm{diag}(\sigma^u)^2)$ for $t=0,\ldots,N-1$.
        \State Set $u_t^{(i)}\gets u_t+\varepsilon_t^{(i)}$ for $t=0,\ldots,N-1$.
        \State Set $x_0^{(i)}\gets x_0$ and $u_{-1}^{(i)}\gets u_{-1}$.
        \For{$t=0,\ldots,N-1$}
            \State $u_t^{(i)}\gets \sat(u_t^{(i)};u_{t-1}^{(i)})$.
            \State $x_{t+1}^{(i)}\gets f(x_t^{(i)},u_t^{(i)})$ and accumulate $J^{(i)}$.
        \EndFor
    \EndFor
    \State Compute $w^{(i)}$ using equation~\eqref{eq:mppi_weights}.
    \For{$t=0,\ldots,N-1$}
        \State $u_t\gets u_t+\dfrac{\sum_{i=1}^{I}w^{(i)}\varepsilon_t^{(i)}}{\sum_{i=1}^{I}w^{(i)}}$.
        \State $u_t\gets \sat(u_t;u_{t-1})$.
    \EndFor
\EndFor
\State Apply $u_0$ to the plant.
\end{algorithmic}
\end{algorithm}

The rollout computations for different samples are independent until the weighted update and can therefore be parallelized on a graphics processing unit.

\subsubsection{Proposed method: MPPI-based optimization of PID gains}
The proposed method optimized the PID-gain vector as follows instead of directly optimizing the control inputs in MPPI:
\begin{align}
    \theta := [& K_{1,1}^\rmP, K_{1,1}^\rmI, K_{1,1}^\rmD, \ldots, K_{n_u,n_{n_u}^e}^\rmP, K_{n_u,n_{n_u}^e}^\rmI, K_{n_u,n_{n_u}^e}^\rmD]^T
    .
\end{align}
Compared with conventional MPPI, the optimization variables were PID gains that remained fixed over the horizon; therefore, the dimension of the decision variables could be significantly reduced.
The resulting algorithm is shown in Algorithm~\ref{alg:mppi_pid}.

\begin{algorithm}[t]
\caption{MPPI--PID for PID-gain optimization}
\label{alg:mppi_pid}
\begin{algorithmic}[1]
\Require Current state $x_0$, previous input $u_{-1}$, current PID-gain vector $\theta$
\For{$k=1,\ldots,N^\mathrm{iter}$}
    \For{$i=1,\ldots,I$}
        \State \textbf{Sample $\varepsilon^{(i)}\sim\calN(0,\mathrm{diag}(\sigma^\theta)^2)$ and set $\theta^{(i)}\gets\theta+\varepsilon^{(i)}$.}
        \State Set $x_0^{(i)}\gets x_0$ and $u_{-1}^{(i)}\gets u_{-1}$.
        \For{$t=0,\ldots,N-1$}
            \State \textbf{Compute $u_t^{(i)}$ from the PID law \eqref{eq:pid} using $\theta^{(i)}$.}
            \State $u_t^{(i)}\gets \sat(u_t^{(i)};u_{t-1}^{(i)})$.
            \State $x_{t+1}^{(i)}\gets f(x_t^{(i)},u_t^{(i)})$ and accumulate $J^{(i)}$.
        \EndFor
    \EndFor
    \State Compute $w^{(i)}$ using equation~\eqref{eq:mppi_weights}.
    \State \textbf{$\theta\gets\theta+\dfrac{\sum_{i=1}^{I}w^{(i)}\varepsilon^{(i)}}{\sum_{i=1}^{I}w^{(i)}}$.}
\EndFor
\State \textbf{Compute $u_0$ from the PID law \eqref{eq:pid} using the updated $\theta$.}
\State Apply $u_0\gets\sat(u_0;u_{-1})$ to the plant.
\end{algorithmic}
\end{algorithm}

The bold statements indicate the steps that differ from conventional MPPI.
As in conventional MPPI, the sample rollouts can be parallelized because each sampled gain vector defines an independent predicted trajectory.

\subsection{Theoretical considerations}
\subsubsection{Information-theoretic interpretation of the algorithm}
In this subsection, we present an information-theoretic interpretation of an update rule that unifies conventional MPPI and the proposed method.
For conventional MPPI, the derivation reduces to the standard information-theoretic path-integral update \cite{mppi_information}. The role of the present formulation was to express both algorithms through a common abstract optimization variable.
This unified expression also facilitated the sampling-efficiency analysis in the next subsection.
Here, the optimization variable was denoted by $z\in\bbR^{n_z}$.
In conventional MPPI and the proposed method, this variable was denoted as $z=[u_0^T, \ldots, u_{N-1}^T]^T$ and $z=\theta$, respectively.
Then, the update rule for $z$ could be derived as in the following proposition.

\begin{proposition}\label{prop:mppi_kl}
Let the probability distribution of Gaussian perturbations for the variable $z$ be $p(\varepsilon):=\calN(\varepsilon; 0, \Sigma^z)$.
Let the objective value computed using the perturbed variable $z+\varepsilon$ be $J(z+\varepsilon)$.
We defined a probability distribution obtained by weighting the perturbation distribution based on this objective value as follows:
\begin{align}
    q(\varepsilon) := \dfrac{1}{Z} p(\varepsilon) \exp\left( -\frac{J(z+\varepsilon)}{\lambda} \right),\
    Z := \int p(\varepsilon) \exp\left( -\frac{J(z+\varepsilon)}{\lambda} \right) d\varepsilon.
\end{align}
Then, we considered the optimization problem of projecting the preceding $q$ onto the set of Gaussian distributions with covariance $\Sigma^z$ in terms of Kullback--Leibler divergence as follows:
\begin{align}
    \min_m D^\mathrm{KL}( q(\varepsilon) || \calN(\varepsilon; m, \Sigma^z) ).
\end{align}
The solution was given by $m^*=\bbE_{q}[\varepsilon]$.
Let the update of the variable $z$ based on this $m^*$ be $z \leftarrow z + m^*$. Here, if a Monte Carlo approximation is applied, it yields the following:
\begin{align}
    z \leftarrow z + \frac{ \sum_{i=1}^I w^{(i)} \varepsilon^{(i)} }{ \sum_{i=1}^I w^{(i)} },\
    w^{(i)} := \exp\left( -\frac{ J(z+\varepsilon^{(i)}) }{ \lambda } \right),\ \varepsilon^{(i)} \sim p(\varepsilon).
    \label{eq:mc_approx}
\end{align}
\end{proposition}
\begin{proof}
First, the definition of Kullback--Leibler divergence provides the following:
\begin{align}
    D^\mathrm{KL}( q(\varepsilon) || \calN(\varepsilon; m, \Sigma^z) )
    &= \int q(\varepsilon) \log \frac{ q(\varepsilon) }{ \calN(\varepsilon; m, \Sigma^z) } d\varepsilon \nonumber \\
    &= \int q(\varepsilon) \log q(\varepsilon) d\varepsilon - \int q(\varepsilon) \log \calN(\varepsilon; m, \Sigma^z) d\varepsilon
    \nonumber \\
    &= \bbE_q[ \log q(\varepsilon) ] - \bbE_q[ \log \calN(\varepsilon; m, \Sigma^z) ] \nonumber \\
    &= \mathrm{const} - \bbE_q\left[ -\dfrac{1}{2} (\varepsilon - m)^T (\Sigma^z)^{-1} (\varepsilon - m) \right]
    \nonumber \\
    &= \mathrm{const} + \dfrac{1}{2}( m - \bbE_q[\varepsilon] )^T (\Sigma^z)^{-1} ( m - \bbE_q[\varepsilon] ).
\end{align}
Therefore, the $m$ that minimizes this expression is given by $m^*=\bbE_q[\varepsilon]$.
Furthermore, the following equation is obtained:
\begin{align}
    \bbE_q[\varepsilon]
    &= \int \varepsilon q(\varepsilon) d\varepsilon
    = \dfrac{1}{Z} \int \varepsilon p(\varepsilon) \exp\left( -\frac{J(z+\varepsilon)}{\lambda} \right) d\varepsilon
    \nonumber \\
    &= \dfrac{ \bbE_{p}[ \varepsilon \exp( -J(z+\varepsilon)/\lambda ) ] }{ \bbE_{p}[ \exp( -J(z+\varepsilon)/\lambda ) ] }.
\end{align}
Thus, the Monte Carlo approximation yields equation~\eqref{eq:mc_approx}.
\end{proof}

Next, we demonstrated that when the objective function is differentiable, the update rule can be interpreted as a gradient descent.

\begin{proposition}
Let us assume that the objective function $J$ is differentiable with respect to $z$ and that the perturbation $\varepsilon$ is sufficiently small.
Then, the distribution $q(\varepsilon)$ in Proposition \ref{prop:mppi_kl} can be approximated by the following:
\begin{align}
    q(\varepsilon) \simeq \calN\left( \varepsilon; - \frac{ \Sigma^z \nabla_z J(z) }{ \lambda }, \Sigma^z \right).
\end{align}
Moreover, the update of the variable $z$ based on this approximation becomes the following:
\begin{align}
    z \leftarrow z - \frac{ \Sigma^z \nabla_z J(z) }{ \lambda }.
\end{align}
\end{proposition}
\begin{proof}
Applying a Taylor expansion to $J$ yields the following: 
\begin{align}
    J(z+\varepsilon) \simeq J(z) + \nabla_z J(z)^T \varepsilon.
\end{align}
Thus, the distribution $q(\varepsilon)$ can be written as the following:
\begin{align}
    q(\varepsilon)
    &\simeq \dfrac{1}{Z} p(\varepsilon) \exp\left( -\frac{ J(z) + \nabla_z J(z)^T \varepsilon }{ \lambda } \right)
    \nonumber \\
    &\propto \exp\left( -\dfrac{1}{2} \varepsilon^T (\Sigma^z)^{-1} \varepsilon - \dfrac{ \nabla_z J(z)^T \varepsilon }{ \lambda } \right) \nonumber \\
    &\propto \exp\left(
        -\dfrac{1}{2} \left( \varepsilon + \dfrac{ \Sigma^z \nabla_z J(z) }{ \lambda } \right)^T (\Sigma^z)^{-1}
        \left( \varepsilon + \dfrac{ \Sigma^z \nabla_z J(z) }{ \lambda } \right)
        \right).
\end{align}
Therefore, $q(\varepsilon)$ can be approximated by $\calN( \varepsilon; - \Sigma^z \nabla_z J(z) / \lambda, \Sigma^z )$.
The latter result follows directly from the definition.
\end{proof}
\subsubsection{Sample efficiency based on effective sample size}
This subsection describes a simplified analysis that illustrates how the dimension of the optimization variable affects the sample efficiency of the MPPI update rule.
As the update in equation~\eqref{eq:mc_approx} corresponds to importance sampling with weights $w^{(i)}$, a standard measure of weight degeneracy and sample efficiency is the effective sample size (ESS) \cite{ess_importance}.
For normalized weights $\widetilde{w}^{(i)}:=w^{(i)}/\sum_{j=1}^I w^{(j)}$, the sample ESS is defined as the following:
\begin{align}
    \widehat{I}^{\mathrm{ESS}}
    := \frac{1}{\sum_{i=1}^{I} (\widetilde{w}^{(i)})^2}
    = \frac{\left(\sum_{i=1}^{I} w^{(i)}\right)^2}{\sum_{i=1}^{I} (w^{(i)})^2}.
\end{align}
Next, we analyzed an idealized counterpart of the preceding quantity, which is given by the following:
\begin{align}
    I^{\mathrm{ESS}}
    := I \frac{(\bbE_{p}[w])^2}{\bbE_{p}[w^2]}.
    \label{eq:ess_ideal}
\end{align}
The evaluation of $I^{\mathrm{ESS}}$ requires the first and second moments of the weight $w$ under $p$.

Under the same differentiability and small-perturbation assumptions as in the previous proposition, we approximated the objective function around $z$ as the following:
\begin{align}
    J(z+\varepsilon) \simeq J(z) + g^T \varepsilon,\
    g := \nabla_z J(z).
\end{align}
Then, the weight in equation~\eqref{eq:mc_approx} becomes the following:
\begin{align}
    w(\varepsilon)
    = \exp\left(-\frac{J(z+\varepsilon)}{\lambda}\right)
    \simeq \exp\left( -\frac{J(z)}{\lambda} \right)
    \exp\left(-\frac{g^T \varepsilon}{\lambda}\right)
    =:c \exp\left(-\frac{g^T \varepsilon}{\lambda}\right).
\end{align}
As $\varepsilon\sim\calN(0,\Sigma^z)$, we have the following:
\begin{align}
    -\frac{g^T \varepsilon}{\lambda}
    \sim \calN\left(0,\ \frac{g^T \Sigma^z g}{\lambda^2}\right),
\end{align}
and hence, $w(\varepsilon)$ follows a log-normal distribution.
Therefore, its moments are given by the following:
\begin{align}
    \bbE_{p}[w]
    = c \exp\left(\frac{1}{2\lambda^2} g^T \Sigma^z g\right),\
    \bbE_{p}[w^2]
    = c^2 \exp\left(\frac{2}{\lambda^2} g^T \Sigma^z g\right).
\end{align}
Substituting these into equation~\eqref{eq:ess_ideal} yields the following:
\begin{align}
    I^{\mathrm{ESS}}
    \simeq I \exp\left(-\frac{1}{\lambda^2} g^T \Sigma^z g\right).
    \label{eq:ess_result_general}
\end{align}

To make the scaling with the dimension explicit, we assumed the following: the perturbation covariance is isotropic, $\Sigma^z = (\sigma^z)^2 I_{n_z}$, and each component contributes similarly to the squared gradient norm.
Under these assumptions, $\|g\|^2 \simeq n_z \overline{g}^2$, and hence $g^T\Sigma^z g \simeq (\sigma^z)^2 n_z \overline{g}^2$.
Equation~\eqref{eq:ess_result_general} becomes the following:
\begin{align}
    I^{\mathrm{ESS}}
    \simeq I \exp\left(
        -\frac{ (\sigma^z)^2 \overline{g}^2 }{\lambda^2} n_z
    \right).
    \label{eq:ess_result_dim}
\end{align}
Equation~\eqref{eq:ess_result_dim} indicates that the effective sample size can decrease exponentially as the optimization dimension $n_z$ increases in this idealized setting.
Therefore, to maintain $I^{\mathrm{ESS}}$ at a desired level, the required number of samples $I$ can grow exponentially with $n_z$.
This argument is not intended to cover cases in which the effective gradient is sparse or otherwise has a special structure that breaks this dimensional scaling.

In conventional MPPI, the optimization variable is the control-input sequence $z=[u_0^T,\ldots,u_{N-1}^T]^T$ and hence, $n_z=n_u N$.
In contrast to conventional MPPI, in the proposed method, $z=\theta$ and $n_z$ equaled the number of PID gains, which was typically small and independent of the horizon length.
Consequently, under the same sampling budget $I$, the proposed method could achieve a larger effective sample size, providing theoretical support for its improved sample efficiency.

\subsubsection{Smoothness of the control input}
In this subsection, we analyze the temporal smoothness of the control input expected from the proposed method.
To provide intuition, we present a simplified analysis of the essential mechanism under a fixed perturbation distribution.

First, the control input based on a PID control can be expressed as the following:
\begin{align}
    u_t(\theta) = u^\mathrm{bias} + E_t \theta
\end{align}
where $E_t$ is a matrix formed by stacking the errors and their derivatives and integrals.
In the proposed method, the PID gains were perturbed from a nominal value as $\theta=\theta^\mathrm{nom}+\varepsilon\ ( \varepsilon\sim\calN(0,\Sigma^\theta) )$; therefore, the perturbation of the control input is the following:
\begin{align}
    \delta u_t := u_t(\theta)-u_t(\theta^\mathrm{nom}) = E_t \varepsilon.
\end{align}
Hence, the covariance between perturbations at different times conditioned on $E$ is the following:
\begin{align}
    \mathrm{Cov}[\delta u_t, \delta u_s | E]
    =\bbE[ E_t \varepsilon\varepsilon^T E_s^T | E]
    = E_t \bbE[\varepsilon\varepsilon^T] E_s^T
    = E_t \Sigma^\theta E_s^T.
\end{align}
This indicates that the control-input perturbations at different times ($t\neq s$) are temporally correlated.
Next, the magnitude of the temporal change in the control-input perturbation is the following:
\begin{align}
    &\bbE[ \| \delta u_t - \delta u_{t-1} \|^2 | E]
    = \bbE[ \| (E_t - E_{t-1})\varepsilon \|^2 | E] \nonumber \\
    &= \bbE[ \varepsilon^T (E_t - E_{t-1})^T (E_t - E_{t-1}) \varepsilon | E]
    =\mathrm{Tr}[ (E_t - E_{t-1}) \Sigma^\theta (E_t - E_{t-1})^T ].
\end{align}
Here, we used the fact that, in the proposed method, the PID gains remain constant throughout the horizon.
The preceding expression indicates that when the error terms vary slowly such that $E_t\simeq E_{t-1}$, the average temporal variation of the control-input perturbation tends to be small.

By contrast, the control input in conventional MPPI can be expressed as the following:
\begin{align}
    u_t = u_t^\mathrm{nom} + \varepsilon_t\ (\varepsilon_t\sim\calN(0,\Sigma^u)).
\end{align}
Thus, the perturbation is $\delta u_t = \varepsilon_t$, and because perturbations at different times are independent, the covariance is $\mathrm{Cov}[\delta u_t, \delta u_s] = 0\ (t \neq s)$.
This indicates that the control-input perturbations have no temporal correlation.
Next, the magnitude of the temporal change in the perturbation is the following:
\begin{align}
    \bbE[ \| \delta u_t - \delta u_{t-1} \|^2 ]
    &= \bbE[ \| \varepsilon_t - \varepsilon_{t-1} \|^2 ]
    = \bbE[ \varepsilon_t^T\varepsilon_t - 2\varepsilon_t^T\varepsilon_{t-1} + \varepsilon_{t-1}^T\varepsilon_{t-1} ]
    \nonumber \\
    &= 2\mathrm{Tr}[ \Sigma^u ].
\end{align}
This indicates that the average variation in the control-input perturbation has a constant magnitude irrespective of the state.
From the preceding analysis, the proposed method is expected to yield smaller input increments and smoother control inputs than conventional MPPI.

\begin{rem}
The aforementioned analysis concerns the control-input sequence generated within the prediction horizon and does not guarantee smooth closed-loop inputs across all control steps.
Nevertheless, smoother predicted input sequences are expected to promote smoother applied inputs in a receding-horizon implementation because the first element of each optimized sequence is applied to the plant.
\end{rem}

\section{Results}
This section comprises the mini-forklift application setup, system identification based on real-machine data, and path-following comparisons among the fixed-gain PID control, conventional MPPI, and proposed MPPI--PID method.

\subsection{Application to forklift path-following control}
\subsubsection{State and control input}
We consider the motion of a forklift on a two-dimensional plane.
Let the state $x_t$ and control input $u_t$ at time $t$ be $x_t = [X_t,\ Y_t,\ s_t,\ c_t,\ v_t^X,\ v_t^Y,\ r_t]^T$, $u_t = [a_t,\ \delta_t]^T$.
Here, $(X,Y)$ denotes the position; $s=\sin(\psi)$ and $c=\cos(\psi)$ are the trigonometric representations of the yaw angle $\psi$; $(v^X,v^Y)=(\dot{X},\dot{Y})$ is the velocity; and $r=\dot{\psi}$ is the yaw rate.
The control input $a$ is the accelerator command, and $\delta$ is the steering angle.

\subsubsection{Physical model}
We used the following physical model for the forklift:
\begin{align}
    f^\mathrm{phys,c} =
    \begin{bmatrix}
        v_t^X \\
        v_t^Y \\
        c_t r_t \\
        -s_t r_t \\
        k^a a_t c_t - k^V v_t^X - r_t v_t^Y \\
        k^a a_t s_t - k^V v_t^Y + r_t v_t^X \\
        k^\delta \sqrt{(v_t^X)^2+(v_t^Y)^2}\ \delta_t - k^r r_t
    \end{bmatrix}
    \label{eq:fork_phys}
\end{align}
Here, $k^a,k^V,k^\delta,k^r$ denote the identification parameters.
Particularly, the dynamics for the velocity $v:=[v^X, v^Y]^T$ in the preceding equation are obtained under the assumption that the speed magnitude $V=\sqrt{(v^X)^2+(v^Y)^2}$ adheres to the following:
\begin{align}
    \dot V(t) = k^a a(t) - k^V V(t).
\end{align}
That is, because $v=V[c, s]^T$, we have the following:
\begin{align}
    \dot v &= \dot V
    \begin{bmatrix}
        c \\
        s
    \end{bmatrix}
    + V
    \begin{bmatrix}
        -s r \\
        c r
    \end{bmatrix}
    = (k^a a - k^V V)
    \begin{bmatrix}
        c \\
        s
    \end{bmatrix}
    + rV
    \begin{bmatrix}
        -s \\
        c
    \end{bmatrix}
    \nonumber \\
    &= k^a a
    \begin{bmatrix}
        c \\
        s
    \end{bmatrix}
    - k^V
    \begin{bmatrix}
        v^X \\
        v^Y
    \end{bmatrix}
    + r
    \begin{bmatrix}
        -v^Y \\
        v^X
    \end{bmatrix}.
\end{align}
This coincides with the velocity dynamics in equation~\eqref{eq:fork_phys}.

\subsubsection{Residual-learning model}
From the derivation in the previous subsection, the physical model analytically describes the position and angle dynamics within the state dimensions, whereas modeling errors are expected in the velocity and angular-velocity dynamics.
Therefore, we set the mask in the residual-learning model to $M^\res=[0,0,0,0,1,1,1]^T$.
Furthermore, the available training data were sparse in the low-speed region, placing the neural network in an extrapolative region.
Thus, we adjusted the contribution of the residual using the following state-dependent weight $w^\res$:
\begin{align}
    w^\res(x):=\frac{(v^X)^2+(v^Y)^2}{(v^X)^2+(v^Y)^2+(V^\mathrm{th})^2}.
    \label{eq:w_res_forklift}
\end{align}
Here, the threshold $V^\mathrm{th}$ is treated as a hyperparameter.
From the preceding equation, when the speed magnitude $V=\sqrt{(v^X)^2+(v^Y)^2}$ is sufficiently small, $w^\res(x)\simeq 0$ and the residual is suppressed.

\begin{rem}
When the state representation includes $s=\sin(\psi)$ and $c=\cos(\psi)$ as per the aforementioned equation, numerical or learning errors may yield prediction results that satisfy $s^2+c^2\neq 1$.
Therefore, during inference, we applied sequential normalization $(s,c) \leftarrow (s,c) / \sqrt{s^2+c^2}$.
\end{rem}
\subsubsection{Errors used in a PID control}
Let the reference path for a path-following control be $p(\xi)\in\R^2$ ($\xi\in\bbR$ is the path parameter).
Let the nearest point on the path be expressed as the following:
\begin{align}
    p_t^* := p(\xi_t^*),\
    \xi_t^* := \mathrm{argmin}_\xi
    \left\|
        \begin{bmatrix}
            X_t \\
            Y_t
        \end{bmatrix} - p(\xi)
    \right\|
\end{align}
Let its tangent- and normal-direction vectors be denoted by $\tau_t^*$ and $n_t^*$, respectively.
Furthermore, let the reference speed be specified as $V_t^\mathrm{ref}$.
Subsequently, we used the following three errors for PID control:
\begin{itemize}
    \item $e_{t,1,1}$: Signed speed error defined as the difference between the vehicle speed $V_t=\sqrt{(v_t^X)^2+(v_t^Y)^2}$ and the reference speed $V_t^\mathrm{ref}$, namely, $V_t^\mathrm{ref} - V_t$.
    \item $e_{t,2,1}$: Signed lateral-position error between the position $[X_t,Y_t]^T$ and the nearest point $p_t^*$, namely, $n_t^{*T}( [X_t,Y_t]^T - p_t^* )$.
    \item $e_{t,2,2}$: Signed angular error between the heading direction $\psi_t$ and the tangent direction $\tau_t^*$ with angle $\psi_t^*$, namely, $\psi_t-\psi_t^*$.
\end{itemize}
In this case, the PID-gain vector is 9-dimensional:
\begin{align}
    \theta=[K_{1,1}^\rmP,K_{1,1}^\rmI,K_{1,1}^\rmD, K_{2,1}^\rmP,K_{2,1}^\rmI,K_{2,1}^\rmD, K_{2,2}^\rmP,K_{2,2}^\rmI,K_{2,2}^\rmD]^T.
\end{align}
\subsubsection{Objective function of MPC}
The stage cost within the objective function for a path-following control is defined as the following:
\begin{align}
    &l_t :=
    w^V\left( \sqrt{ (v_t^X)^2+(v_t^Y)^2 }-V_t^\mathrm{ref} \right)^2
    + w^\mathrm{path}\left\|
        \begin{bmatrix}
            X_t \\
            Y_t
        \end{bmatrix} - p_t^*
    \right\|^2
    \nonumber \\
    &\hspace{10mm}
    + w^\mathrm{align}\left(
        1-
        \begin{bmatrix}
            c_t \\
            s_t
        \end{bmatrix}^T \tau_t^*
    \right)
    + w^{\Delta u} \norm{u_t-u_{t-1}}^2
    + w^\mathrm{goal} l_t^\mathrm{goal}.
\end{align}
Here, each $w$ denotes a weight parameter, while $l_t^\mathrm{goal}$ represents a penalty term that enforces a speed restriction near the goal, defined as the following:
\begin{align}
    l_t^\mathrm{goal} :=
    \begin{cases}
        \left(
            \sqrt{(v_t^X)^2+(v_t^Y)^2}
            - \varepsilon^{\mathrm{goal,vel}}
        \right)^2
        & \text{if }
        \begin{aligned}[t]
            &\left\|
                \begin{bmatrix}
                    X_t \\
                    Y_t
                \end{bmatrix}
                - p^\mathrm{goal}
            \right\|
            \le \varepsilon^{\mathrm{goal,pos}} \\
            &\text{and }
            \sqrt{(v_t^X)^2+(v_t^Y)^2}
            \ge \varepsilon^{\mathrm{goal,vel}}
        \end{aligned}
        \\
        0 & \text{otherwise}.
    \end{cases}
\end{align}
Here, $p^\mathrm{goal}\in\bbR^2$ denotes the goal position, and $\varepsilon^{\mathrm{goal,pos}}$ and $\varepsilon^{\mathrm{goal,vel}}$ represent the thresholds.
This objective function is non-differentiable, and therefore, gradient-based optimization methods cannot be directly applied.

\subsection{Dataset construction}
In this study, we constructed a time-series dataset consisting of the state and control inputs that are $x_t=[X_t,Y_t,s_t,c_t,v_t^X,v_t^Y,r_t]^T$ and $u_t=[a_t,\delta_t]^T$, respectively, which are derived from manual real-machine driving data of the mini forklift system introduced in \cite{fork_rl}.
The original dataset contained $[X_t,Y_t,\psi_t,v_t^X,v_t^Y,r_t,a_t,\delta_t]^T$; however, in this work, we used $s_t=\sin(\psi_t)$ and $c_t=\cos(\psi_t)$ to avoid wraparound jumps of the yaw angle $\psi$.

The data were recorded at a sampling rate of approximately 15 Hz ($h\simeq 1/15\simeq 0.0667$ s), yielding a total of 70{,}926 samples corresponding to approximately 78.8 minutes.
Owing to missing intervals in the data, the log was segmented at points where the time difference exceeded 0.1 s.
Subsequently, for each continuous interval, the data were resampled via linear interpolation to yield an equally spaced time series ($h=0.0667$ s), and then, a median filter and moving-average filter were applied.
The filter window sizes were set for each signal as follows:
\begin{itemize}
  \item $X,Y,\psi$: median window 3, moving-average window 9.
  \item $v^X,v^Y,r$: median window 3, moving-average window 7.
  \item $a,\delta$: median window 3, moving-average window 5.
\end{itemize}
Finally, to align the basic unit used in training and evaluation, each preprocessed continuous interval was uniformly partitioned into 5-s segments, which resulted in approximately 75 steps.
Further, these fixed-length segments were split into training, validation, and test datasets in a 70:15:15 ratio.

\subsection{System identification results}
The identification parameters $k^a,k^V,k^\delta,k^r$ of the physical model were estimated via least squares using the full training dataset.
For both the standard neural network and residual-learning model, the state and control input were normalized by dividing by their per-dimension standard deviation.
Minibatch learning was employed to minimize the one-step-ahead state-prediction loss.
The main settings for learning are summarized in Table~\ref{tab:setting_id}.
\begin{table}[t]
  \centering
  \small
  \tabcaption{Main settings for learning-based system identification}
  \label{tab:setting_id}
  \begin{tabular}{p{0.34\linewidth}p{0.60\linewidth}}
    \hline
    Item & Setting \\
    \hline
    Neural network architecture & 3 layers, 50 units, activation function: ReLU \\
    Training settings & Optimizer: Adam, learning rate $1\times 10^{-4}$, batch size 1024, gradient clipping threshold 1.0, standard deviation of noise added to next state $1\times 10^{-4}$, early stopping patience 10 epochs \\
    Parameter of residual weight $w^\res$ & $V^\mathrm{th}=0.20$ m/s \\
    \hline
  \end{tabular}
\end{table}

Using the physical model, standard neural-network model, and  residual-learning model trained by the preceding method, we evaluated the recursive-state-prediction accuracy on the fixed-length test segments using the coefficient of determination $R^2$.
The results are summarized in Table~\ref{tab:r2}.
\begin{table}[t]
  \centering
  \small
  \tabcaption{Coefficient of determination $R^2$ for recursive-state prediction over 5.0 s}
  \label{tab:r2}
  \resizebox{\linewidth}{!}{%
    \begin{tabular}{lcccccccc}
        \hline
        Model & $X$ & $Y$ & $s$ & $c$ & $v^X$ & $v^Y$ & $r$ & Average\\
        \hline
        Physical model & 0.9842 & 0.9926 & 0.9625 & 0.9455 & 0.7288 & 0.8380 & 0.9102 & 0.9088\\
        Standard neural network & 0.9855 & 0.9774 & 0.9684 & 0.9555 & 0.8394 & 0.9089 & \textbf{0.9465} & 0.9402\\
        Residual-learning model & \textbf{0.9920} & \textbf{0.9965} & \textbf{0.9759} & \textbf{0.9648} & \textbf{0.8585} & \textbf{0.9251} & 0.9396 & \textbf{0.9503}\\
        \hline
    \end{tabular}}
\end{table}
Table~\ref{tab:r2} indicates that the residual-learning model yields the highest average $R^2$.
This model achieves superior $R^2$ values in six of the seven state components, while the yaw rate $r$ is best captured by the standard neural-network model.
Based on the highest overall accuracy, the residual-learning model was adopted in the subsequent path-following control study.

\subsection{Path-following control results}
Path-following control was implemented as an MPC that sequentially computed control-input sequences based on the residual-learning model identified in the previous section.
The reference path was generated from the start and goal states using a cubic Hermite spline.
The number of path points was set to 800.
The parameters used for the control are summarized in Table~\ref{tab:setting_ctrl}.
\begin{table}[t]
  \centering
  \small
  \tabcaption{Parameters used in the path-following control}
  \label{tab:setting_ctrl}
  \begin{tabular}{p{0.34\linewidth}p{0.60\linewidth}}
    \hline
    Item & Setting \\
    \hline
    Sampling period & $h=0.0667$ s ($\simeq$ 15 Hz)\\
    Horizon length & $N=60$ ($\simeq$ 4.0 s)\\
    Number of samples for single-run comparisons & $I=2048$ (baseline), $I=16$ (low-sample case)\\
    Number of iterations & $N^\mathrm{iter}=3$\\
    Temperature parameter & $\lambda=1.0$ \\
    Input bounds & $0\leq a\leq 100$, $-65\leq \delta\leq 65$ deg \\
    Input-rate bounds & $|\dot{a}|\leq 32.0$/s, $|\dot{\delta}|\leq 100$ deg/s \\
    One-step change bounds & $|\Delta a|\leq 32.0h\simeq 2.13$, $|\Delta\delta|\leq 100h\simeq 6.67$ deg \\
    Noise-standard deviation in conventional MPPI & $\sigma^u=[8.0,\ 6.0]^T$ (corresponding to $a,\delta$)\\
    Bias in PID control input & $u^\mathrm{bias}=[40.0,\ 0.0]^T$ (corresponding to $a,\delta$) \\
    Initial PID gains & $[K_{1,1}^\rmP,K_{1,1}^\rmI,K_{1,1}^\rmD,K_{2,1}^\rmP,K_{2,1}^\rmI,K_{2,1}^\rmD,K_{2,2}^\rmP,K_{2,2}^\rmI,K_{2,2}^\rmD]$ \\
    & $=[50.0, 0.20, 10.0, 100.0, 0.10, 5.0, 150.0, 0.10, 3.0]$ \\
    Noise-standard deviation of PID gains & $[10.0,\ 0.10,\ 5.0,\ 30.0,\ 0.05,\ 2.0,\ 20.0,\ 0.05,\ 2.0]$ \\
    Cost weights & $w^V=50$, $w^\mathrm{path}=500$, $w^\mathrm{align}=10$, $w^{\Delta u}=0.05$, $w^\mathrm{goal}=50000$ \\
    Thresholds near the goal & $\varepsilon^{\mathrm{goal,pos}}=0.10$ m, $\varepsilon^{\mathrm{goal,vel}}=0.04$ m/s \\
    Reference speed & $V_t^\mathrm{ref}=0.10$ m/s \\
    \hline
  \end{tabular}
\end{table}
In this section, we compared the fixed-gain PID control, conventional MPPI, and proposed method (MPPI--PID).
The fixed-gain PID control fixed the PID gains at the initial values listed in Table~\ref{tab:setting_ctrl} and computed the control input solely according to the PID law \eqref{eq:pid}, without online optimization.
For both conventional MPPI and the proposed method, we employed the same prediction model, horizon length, cost function, and number of samples $I$ for each setting, ensuring that only the difference in methods was reflected in the results.

\subsubsection{Single-run path-following comparisons}
In this section, we initially explain the baseline results with $I=2048$ and then the results with a substantially reduced number of samples $I=16$, with all other parameters unchanged.
Figure~\ref{fig:traj} depicts the planar trajectories generated under the fixed-gain PID control, conventional MPPI, and  proposed method (MPPI--PID) in the baseline setting ($I=2048$).
\begin{figure}[!htbp]
  \centering
  \includegraphics[width=0.99\linewidth]{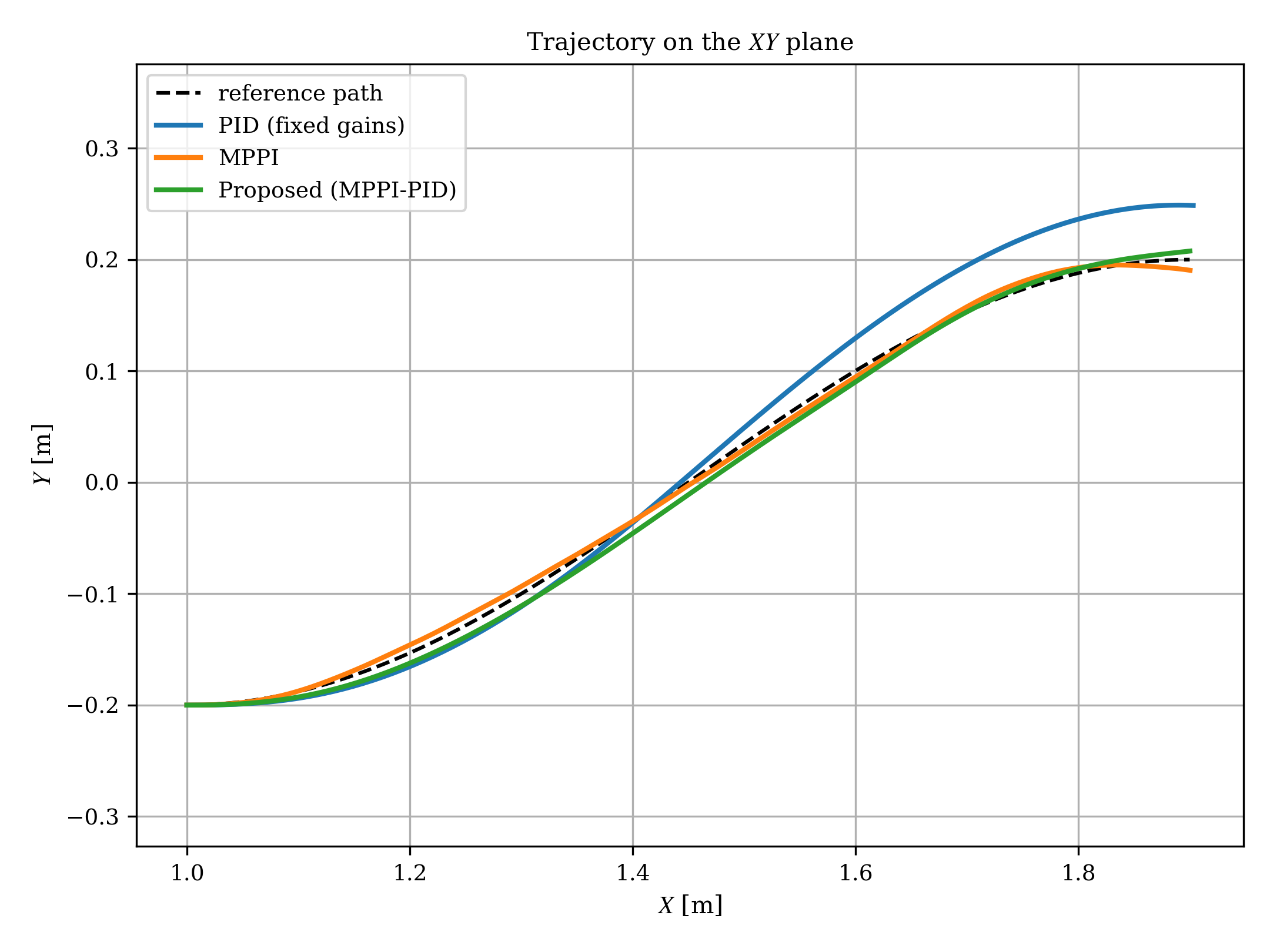}
  \caption{Planar trajectories obtained in the path-following control for $I=2048$}
  \label{fig:traj}
\end{figure}
In this figure, the dashed line represents the reference path, whereas the solid lines denote the tracking trajectories obtained by each method.
Although conventional MPPI and the proposed method exhibit approximately comparable tracking performance relative to the reference path, the fixed-gain PID control demonstrates a larger deviation from the reference path in the latter part.

Next, Figure~\ref{fig:state_input} presents the time-series of the states, control inputs, path-following errors, and one-step control-input increments obtained under the fixed-gain PID control, conventional MPPI, and MPPI--PID in the baseline setting ($I=2048$).
\begin{figure}[!htbp]
  \centering
  \includegraphics[width=0.99\linewidth]{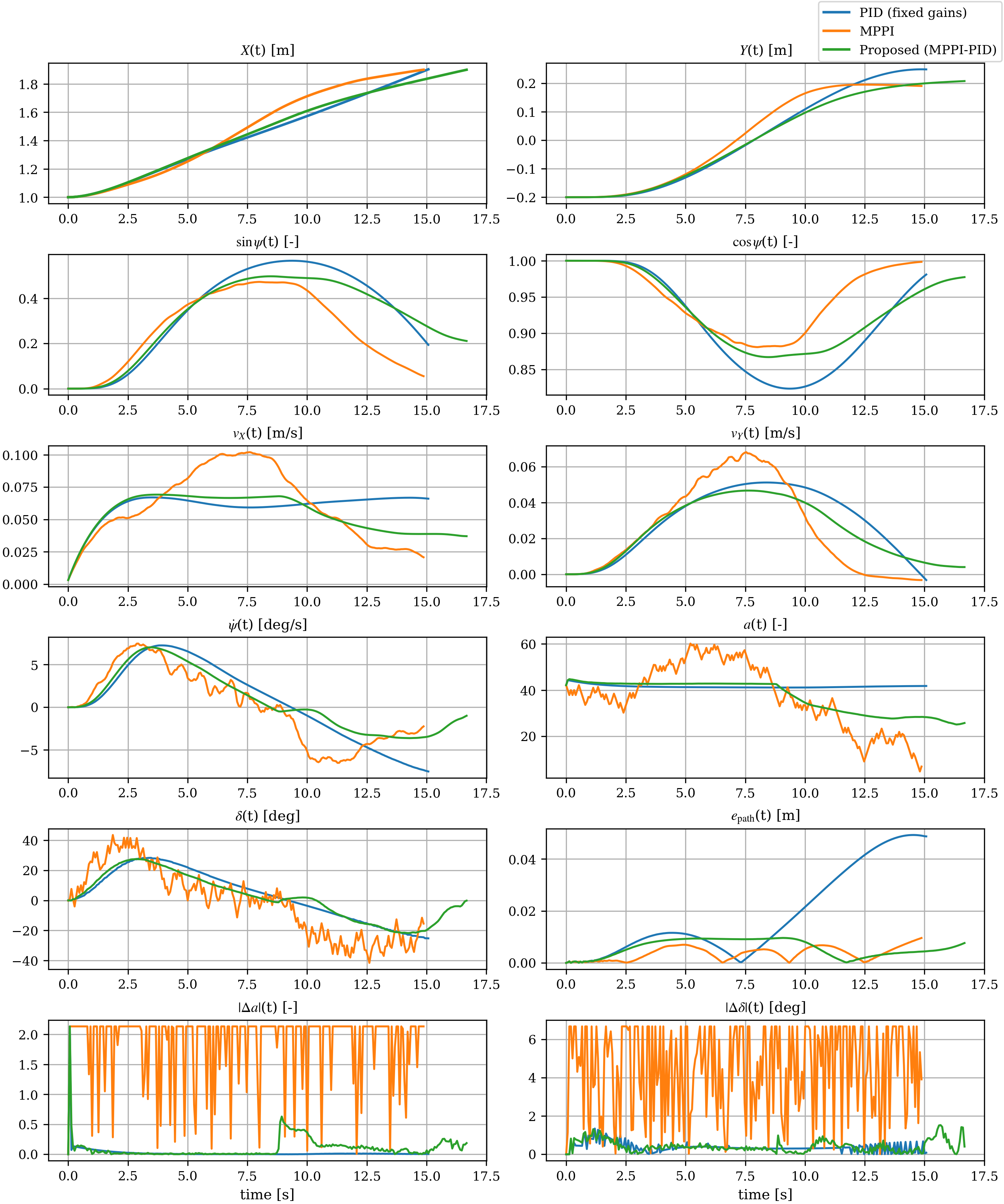}
  \caption{Time-series of the states, control inputs, path-following errors, and control-input increments for $I=2048$}
  \label{fig:state_input}
\end{figure}
In this figure, the path-following errors and one-step control-input increments are defined as $e_t^\mathrm{path}:= \| [X_t,Y_t]^T - p_t^* \|$, $|\Delta a_t|:=|a_t-a_{t-1}|$, and $|\Delta\delta_t|:=|\delta_t-\delta_{t-1}|$, respectively.
In conventional MPPI, $|\Delta a_t|$ and $|\Delta\delta_t|$ frequently approach the upper bounds of the one-step change constraints in Table~\ref{tab:setting_ctrl}, resulting in large fluctuations of the control inputs.
In MPPI--PID, the temporal changes of $a$ and $\delta$ are smaller than those in conventional MPPI, while the path error $e_t^\mathrm{path}$ remains at a level comparable to that of conventional MPPI.
The fixed-gain PID control yields smooth control inputs, but $e_t^\mathrm{path}$ increases in the latter part.

Figure~\ref{fig:traj_16} illustrates the planar trajectories for $I=16$.
\begin{figure}[!htbp]
  \centering
  \includegraphics[width=0.99\linewidth]{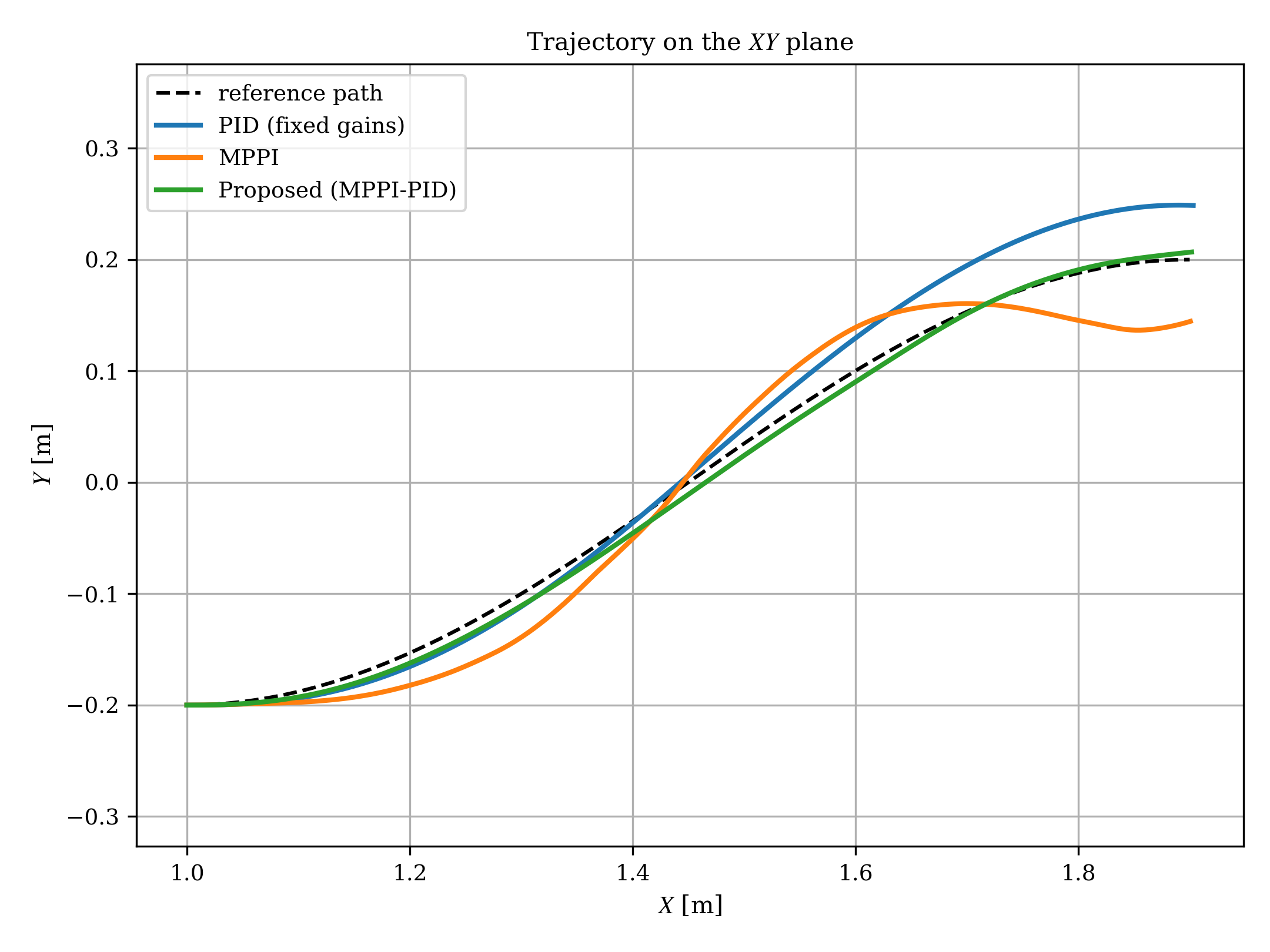}
  \caption{Planar trajectories obtained in the path-following control for $I=16$}
  \label{fig:traj_16}
\end{figure}
Compared with the baseline setting ($I=2048$), conventional MPPI exhibits a larger deviation from the reference path when the sample size is reduced to $I=16$, whereas MPPI--PID preserves a tracking trajectory closely aligned with the reference path.
The fixed-gain PID control still exhibits a larger deviation in the latter part.

Figure~\ref{fig:state_input_16} depicts the corresponding time-series for $I=16$.
\begin{figure}[!htbp]
  \centering
  \includegraphics[width=0.99\linewidth]{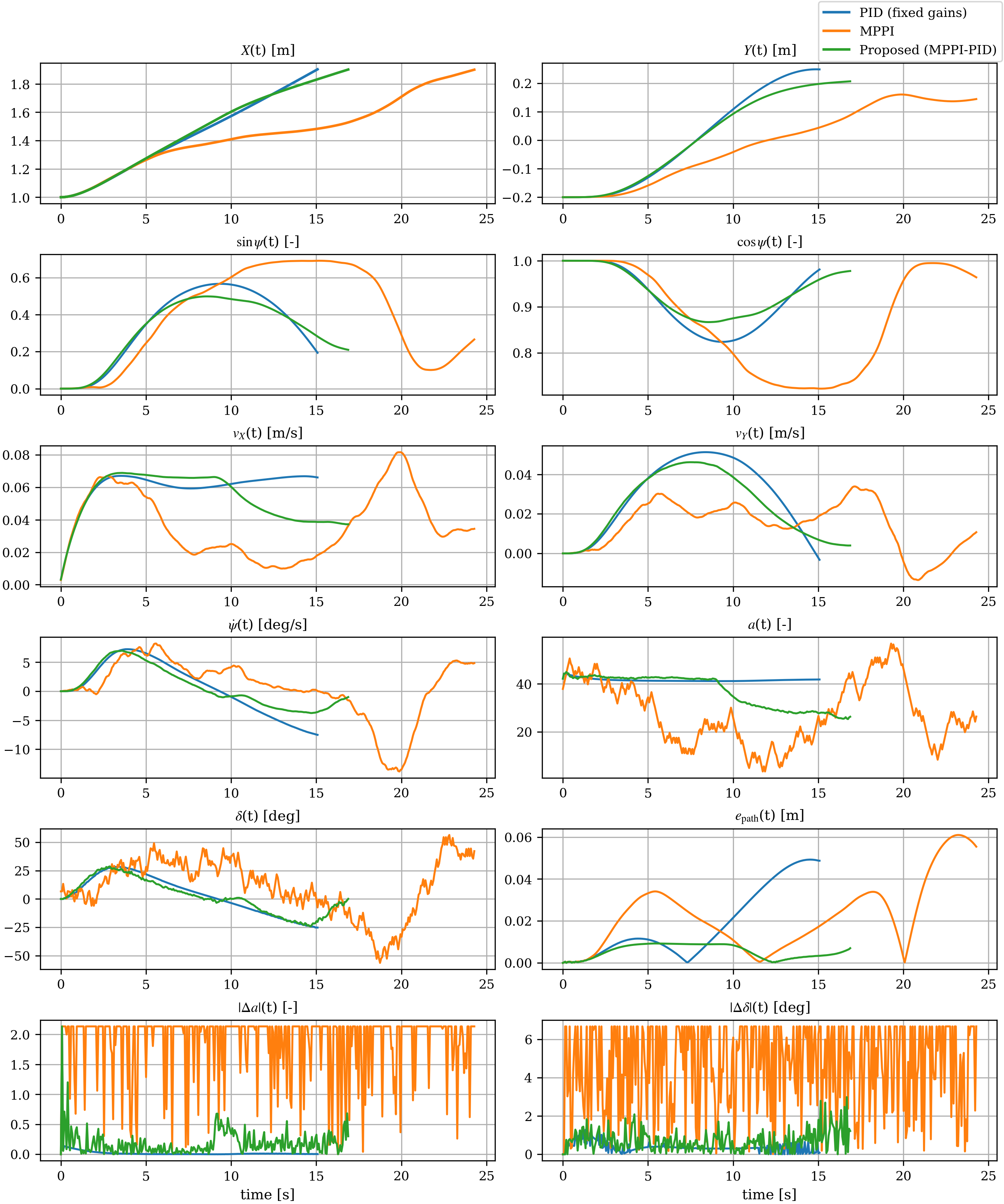}
  \caption{Time-series of the states, control inputs, path-following errors, and control-input increments for $I=16$}
  \label{fig:state_input_16}
\end{figure}
From the perspective of the path error $e_t^\mathrm{path}$, conventional MPPI exhibits larger errors under the low-sample setting, while MPPI--PID maintains smaller errors.
Moreover, the control-input increments $|\Delta a_t|$ and $|\Delta\delta_t|$ in conventional MPPI frequently approach the one-step change bounds, whereas MPPI--PID keeps the increments smaller.

\subsubsection{Evaluation of the sample efficiency, perturbation sensitivity, and computation time}
Next, we evaluated conventional MPPI and MPPI--PID across multiple sample counts, perturbation standard deviations, and random seeds.
The settings are summarized in Table~\ref{tab:auto_experiment_setting}.
The perturbation standard deviations are obtained by multiplying the nominal values listed in Table~\ref{tab:setting_ctrl} by 0.5, 1, or 2.

\begin{table}[t]
  \centering
  \small
  \tabcaption{Evaluation settings for the sample efficiency, perturbation sensitivity, and computation time}
  \label{tab:auto_experiment_setting}
  \begin{tabular}{p{0.40\linewidth}p{0.54\linewidth}}
    \hline
    Item & Setting \\
    \hline
    Methods & Conventional MPPI, MPPI--PID \\
    Sample counts & $I=16$, 256, 2048 \\
    Multipliers for nominal perturbation standard deviations & 0.5, 1, 2 \\
    Repeated random trials & 5 random seeds \\
    Evaluated metrics & Mean path error, mean control-input increment $\|\Delta u_t\|_2$, mean computation time per control step \\
    \hline
  \end{tabular}
\end{table}

Figure~\ref{fig:auto_path_error} shows the mean path error under each condition.
\begin{figure}[!htbp]
  \centering
  \includegraphics[width=0.99\linewidth]{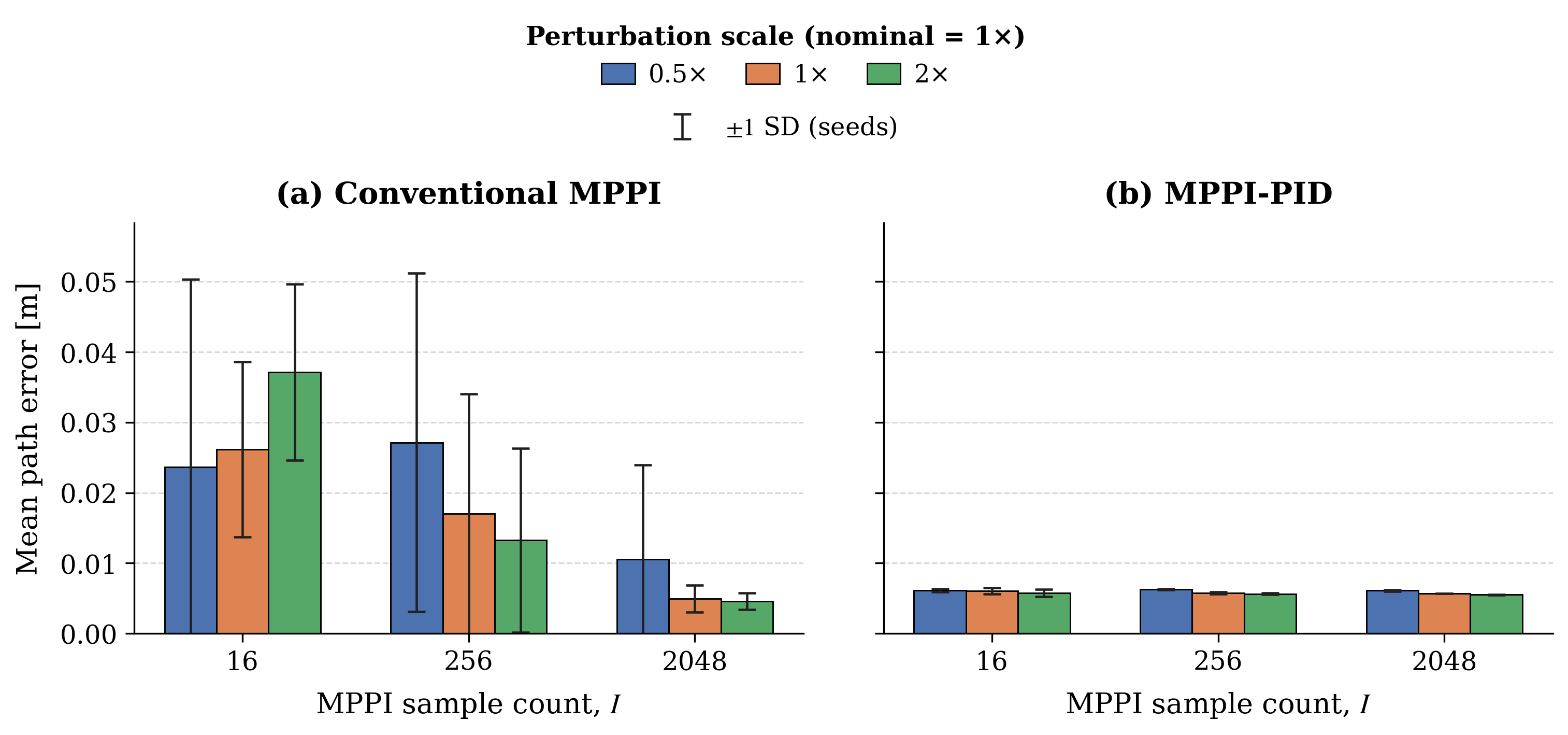}
  \caption{Mean path error for different sample counts and perturbation standard deviations}
  \label{fig:auto_path_error}
\end{figure}
Across all tested sample counts and perturbation standard deviations, MPPI--PID maintains both a small mean path error and low variation across random trials.
By contrast, conventional MPPI exhibits larger path errors and higher trial-to-trial variation, particularly when the sample count is small.

Figure~\ref{fig:auto_control_increment} shows the mean control-input increment $\|\Delta u_t\|_2$.
\begin{figure}[!htbp]
  \centering
  \includegraphics[width=0.99\linewidth]{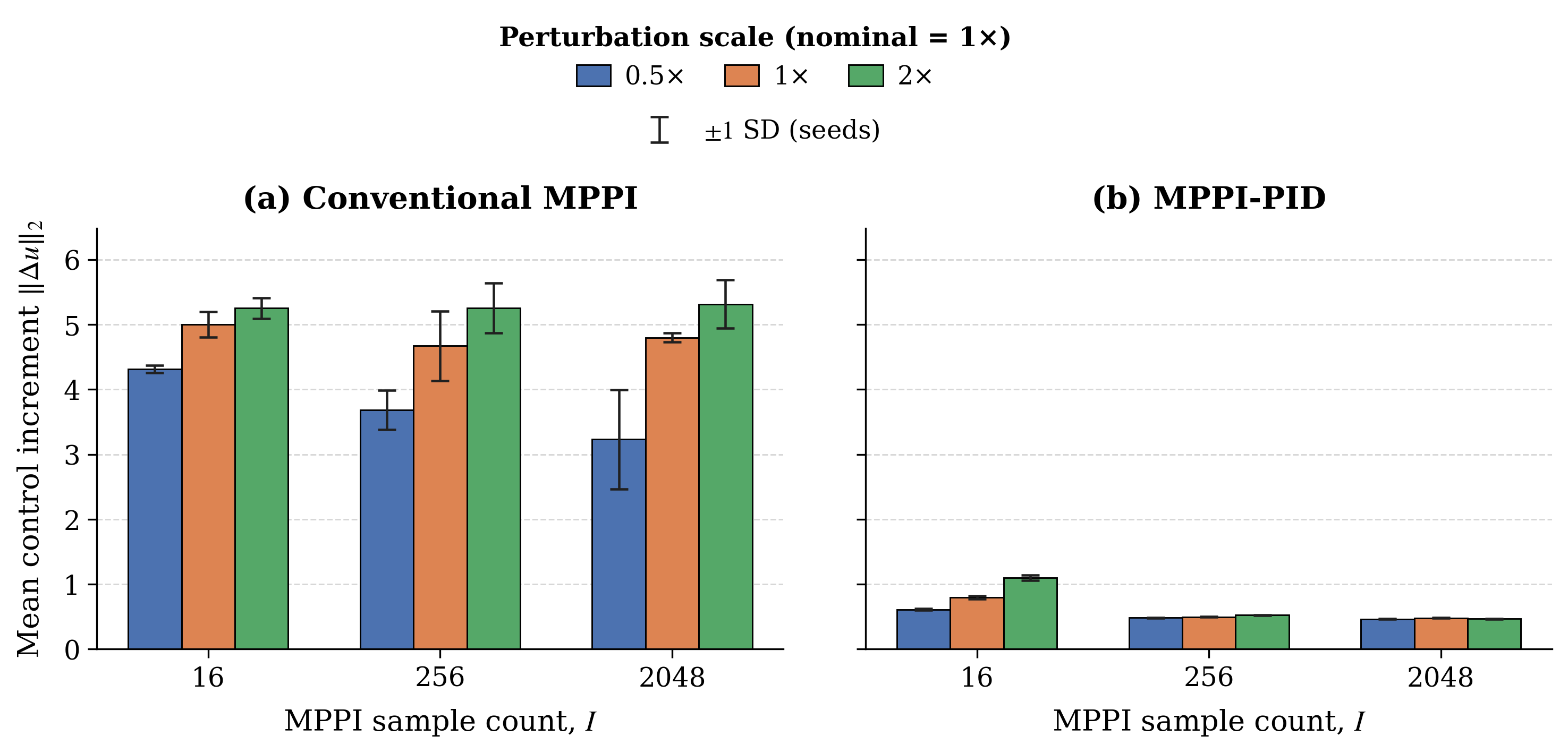}
  \caption{Mean control-input increment for different sample counts and perturbation standard deviations}
  \label{fig:auto_control_increment}
\end{figure}
For conventional MPPI, a larger perturbation standard deviation tends to produce larger control-input increments, whereas changing the sample count has a smaller effect.
MPPI--PID maintains both small control-input increments and low trial-to-trial variation across random trials for all tested conditions.
These results agree with the temporal-correlation analysis of the PID-induced input perturbations.

Figure~\ref{fig:auto_computation_time} shows the mean computation time per control step.
\begin{figure}[!htbp]
  \centering
  \includegraphics[width=0.99\linewidth]{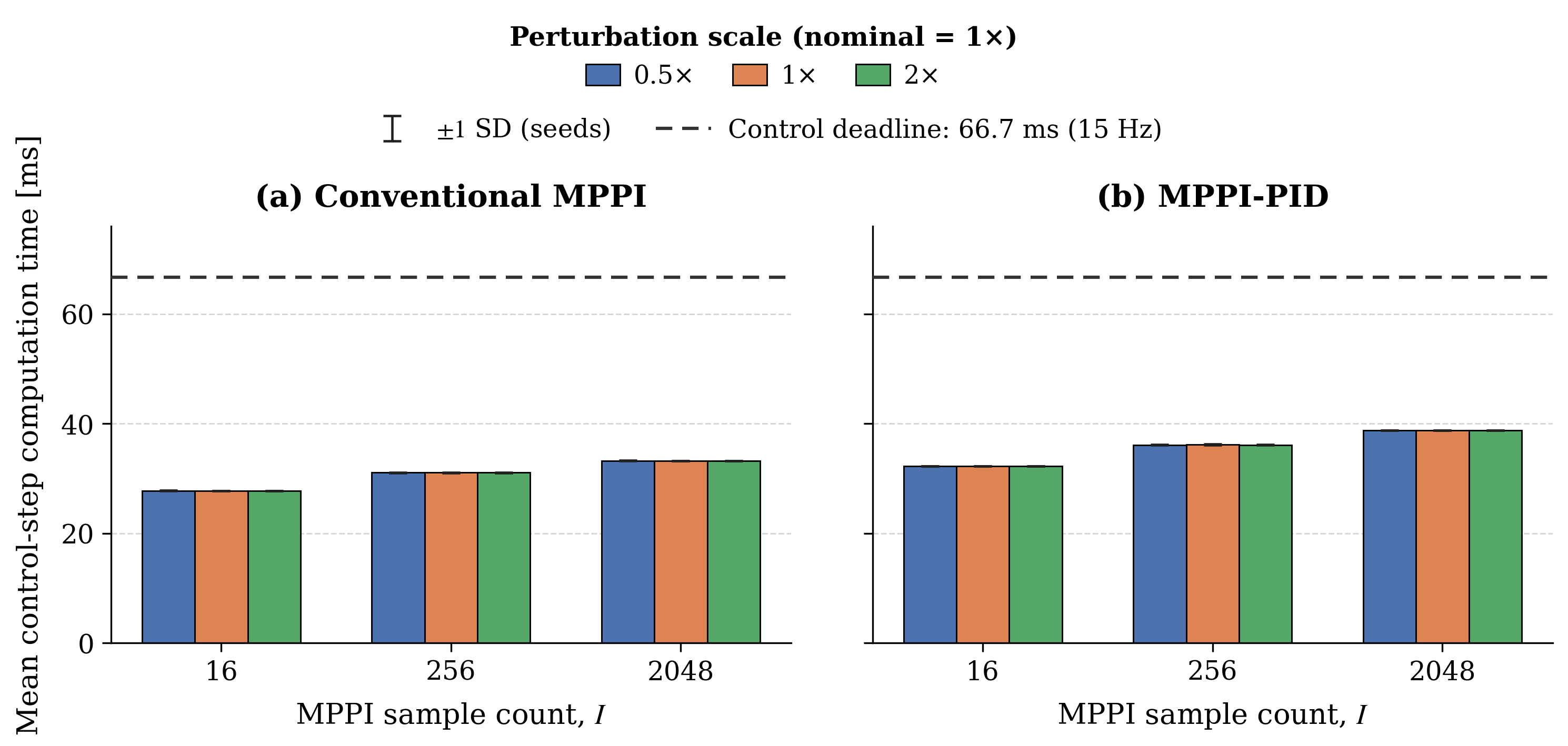}
  \caption{Mean computation time per control step for different sample counts and perturbation standard deviations}
  \label{fig:auto_computation_time}
\end{figure}
The computation time increases with the number of samples and is largely unaffected by the perturbation standard deviation and random trial.
MPPI--PID requires slightly more computation time than conventional MPPI because it computes the control-input sequence from sampled PID gains, but all tested conditions remain below the 66.7-ms control deadline corresponding to a 15-Hz operation.

\section{Discussion}
From Table~\ref{tab:r2}, the residual-learning model attains the highest average $R^2$.
Particularly, the velocities $(v^X,v^Y)$ exhibit marked improvement relative to the physical-based model, indicating that the residual neural network corrects nonlinear dynamics that are not fully captured by the physical model.
The position $(X,Y)$ and attitude $(s,c)$ already achieve high values at the physical-model stage, yet the residual-learning model slightly improves them further.
An exception arises in the yaw rate $r$, where the standard neural-network model performs best.
As position and speed errors directly influence the control performance in the path-following control, the residual-learning model, which achieved the highest overall accuracy, was adopted for control.

The control results are consistent with the theoretical considerations.
The effective-sample-size analysis suggests that reducing the optimization dimension can improve sampling efficiency.
In the present path-following setting, conventional MPPI optimizes an input sequence of dimension $n_u N = 2 \times 60 = 120$, whereas MPPI--PID optimizes the 9-dimensional PID-gain vector $\theta$.
This dimensional reduction corresponds to the low-sample and multi-condition results, where MPPI--PID maintains small path errors with fewer samples than conventional MPPI.
The temporal-correlation analysis suggests that sampling PID gains tends to produce smoother predicted-input sequences and smaller temporal-input increments.
This is reflected in the single-run time-series and broader assessment of control-input increments, where MPPI--PID suppresses input increments over a range of sample counts and perturbation standard deviations.

The fixed-gain PID control yields smooth control inputs but shows increased path deviation in the latter part of the reference path.
This result indicates the benefit of online gain optimization: MPPI--PID retains the PID structure while adapting the gains according to predicted future costs.
The computation-time results also indicate that the additional calculation required for generating inputs from sampled gains remains within the control period in the tested implementation.

\section{Conclusions}
This study proposed MPPI--PID, a receding-horizon PID-gain optimization method based on MPPI.
This method optimized PID gains rather than control-input sequences, thereby retaining the PID structure while exploiting gradient-free predictive optimization.
A unified path-integral update, an effective-sample-size analysis, and an input-smoothness analysis endorsed the expected advantages in sample efficiency and control-input smoothness.
In the learning-based mini-forklift path-following task, MPPI--PID improved tracking performance over fixed-gain PID, reduced input increments compared with conventional MPPI, and maintained favorable performance under reduced sample counts.

The scope for future work includes experimental validation on a real machine, extension to stochastic control under model uncertainty, and combination with post-filters such as safety filters for obstacle avoidance and low-pass filters.

\section*{CRediT authorship contribution statement}
Teruki Kato: Conceptualization, Methodology, Software, Validation, Formal analysis, Investigation, Data curation, Writing -- original draft, Writing -- review and editing, Visualization.
Koshi Oishi: Conceptualization, Resources, Investigation, Data curation, Writing -- review and editing, Supervision.
Seigo Ito: Conceptualization, Resources, Writing -- review and editing, Supervision, Project administration.

\section*{Declaration of competing interest}
The authors declare that they have no known competing financial interests or personal relationships that could have appeared to influence the work reported in this paper.

\section*{Acknowledgments}
This research did not receive any specific grant from funding agencies in the public, commercial, or not-for-profit sectors.

\section*{Data availability}
The data that support the findings of this study are not publicly available but may be available from the corresponding author upon reasonable request.

\section*{Appendix: Proof of the equivalence between the saturation mapping and projection}
\begin{lemma}
The saturation mapping \eqref{eq:saturation_map} coincides with the projection\linebreak $\Pi_{\calU(u_{t-1})}(u_t)$.
\begin{proof}
First, from the property of the clipping function, the following holds:
\begin{align}
    u_{t-1} + \varphi_{ [\underline{\Delta u}, \overline{\Delta u}] }^\clip(u_t - u_{t-1})
    = \varphi_{ [u_{t-1} + \underline{\Delta u},\ u_{t-1} + \overline{\Delta u}] }^\clip(u_t)
\end{align}

Therefore, the right-hand side of equation~\eqref{eq:saturation_map} becomes the following:
\begin{align}
    &\varphi_{ [\underline{u}, \overline{u}] }^\clip
    \left(
        u_{t-1} + \varphi_{ [\underline{\Delta u}, \overline{\Delta u}] }^\clip(u_t - u_{t-1})
    \right)
    = \varphi_{ [\underline{u}, \overline{u}] }^\clip
    \left(
        \varphi_{ [u_{t-1} + \underline{\Delta u},\ u_{t-1} + \overline{\Delta u}] }^\clip(u_t)
    \right) \nonumber \\
    &= \varphi_{ [ L(u_{t-1}),U(u_{t-1}) ] }^\clip(u_t)
    = \Pi_{ [L(u_{t-1}),U(u_{t-1})] }(u_t).
\end{align}
Here, we define $L(u_{t-1}):=\max(\underline{u}, u_{t-1} + \underline{\Delta u})$ and $U(u_{t-1}):=\min(\overline{u}, u_{t-1} + \overline{\Delta u})$.
We also use the equivalence between the clipping operation and projection.
The control-input constraint set $\calU(u_{t-1})$ can be written as the following:
\begin{align}
    &\calU(u_{t-1}) = \{ u\in\bbR^{n_u} | \underline{u} \leq u \leq \overline{u},\ \underline{\Delta u} \leq u - u_{t-1} \leq \overline{\Delta u} \}
    \nonumber \\
    &=\{ u\in\bbR^{n_u} | \underline{u} \leq u \leq \overline{u},\ 
    u_{t-1}+\underline{\Delta u} \leq u \leq u_{t-1}+\overline{\Delta u} \}
    \nonumber \\
    &=\{ u\in\bbR^{n_u} | L(u_{t-1}) \leq u \leq U(u_{t-1}) \}
    = [ L(u_{t-1}), U(u_{t-1}) ].
\end{align}
Thus, the right-hand side of equation~\eqref{eq:saturation_map} coincides with the required projection.
\end{proof}
\end{lemma}
\bibliographystyle{elsarticle-num}
\bibliography{references}

@article{ess_importance,
  author  = {Martino, Luca and Elvira, V{\'i}ctor and Louzada, Francisco},
  title   = {Effective sample size for importance sampling based on discrepancy measures},
  journal = {Signal Process.},
  year    = {2017},
  volume  = {131},
  pages   = {386--401},
  doi     = {10.1016/j.sigpro.2016.08.025}
}

@book{mpc_handbook,
  editor    = {Rakovi{\'c}, Sa{\v s}a V. and Levine, William S.},
  title     = {Handbook of Model Predictive Control},
  publisher = {Birkh{\"a}user},
  address   = {Cham},
  series    = {Control Engineering},
  year      = {2018},
  doi       = {10.1007/978-3-319-77489-3}
}

@book{brunton_kutz,
  author    = {Brunton, Steven L. and Kutz, J. Nathan},
  title     = {Data-Driven Science and Engineering: Machine Learning, Dynamical Systems, and Control},
  publisher = {Cambridge University Press},
  address   = {Cambridge},
  year      = {2019},
  doi       = {10.1017/9781108380690}
}

@article{mppi_information,
  author  = {Williams, Grady and Drews, Paul and Goldfain, Brian and Rehg, James M. and Theodorou, Evangelos A.},
  title   = {Information-theoretic model predictive control: theory and applications to autonomous driving},
  journal = {IEEE Trans. Robot.},
  year    = {2018},
  volume  = {34},
  number  = {6},
  pages   = {1603--1622},
  doi     = {10.1109/TRO.2018.2865891}
}

@inproceedings{mppi_stl,
  author    = {Halder, Patrick and Homburger, Hannes and Kiltz, Lothar and Reuter, Johannes and Althoff, Matthias},
  title     = {Trajectory planning with signal temporal logic costs using deterministic path integral optimization},
  booktitle = {2025 IEEE International Conference on Robotics and Automation (ICRA)},
  year      = {2025},
  pages     = {4221--4228},
  doi       = {10.1109/ICRA55743.2025.11127582}
}

@article{pi2,
  author  = {Theodorou, Evangelos A. and Buchli, Jonas and Schaal, Stefan},
  title   = {A generalized path integral control approach to reinforcement learning},
  journal = {J. Mach. Learn. Res.},
  year    = {2010},
  volume  = {11},
  number  = {104},
  pages   = {3137--3181}
}

@inproceedings{pi2_pid,
  author    = {Buchli, Jonas and Theodorou, Evangelos A. and Stulp, Freek and Schaal, Stefan},
  title     = {Variable impedance control -- a reinforcement learning approach},
  booktitle = {Proceedings of Robotics: Science and Systems},
  address   = {Zaragoza, Spain},
  year      = {2010},
  doi       = {10.15607/RSS.2010.VI.020}
}

@article{pid_mpc_1,
  author  = {Xu, Min and Li, Shaoyuan and Qi, Chenkun and Cai, Wenjian},
  title   = {Auto-tuning of {PID} controller parameters with supervised receding horizon optimization},
  journal = {ISA Trans.},
  year    = {2005},
  volume  = {44},
  number  = {4},
  pages   = {491--500},
  doi     = {10.1016/S0019-0578(07)60056-1}
}

@article{pid_mpc_2,
  author  = {Wu, Yongling and Li, Shaoyuan and Li, Kang},
  title   = {Enhanced receding horizon optimal performance for online tuning of {PID} controller parameters},
  journal = {Int. J. Model. Identif. Control},
  year    = {2018},
  volume  = {29},
  number  = {3},
  pages   = {209--217},
  doi     = {10.1504/IJMIC.2018.091239}
}

@inproceedings{pid_mpc_3,
  author    = {Gashi, Fatos and Abuibaid, Khalil and Ruskowski, Martin and Wagner, Achim},
  title     = {Model predictive control based reference generation for optimal proportional integral derivative control},
  booktitle = {2024 32nd Mediterranean Conference on Control and Automation (MED)},
  year      = {2024},
  pages     = {518--524},
  doi       = {10.1109/MED61351.2024.10566273}
}

@misc{pinn_pid,
  author       = {Ito, Junsei and Wasa, Yasuaki},
  title        = {Data-driven adaptive {PID} control based on physics-informed neural networks},
  year         = {2025},
  howpublished = {arXiv preprint arXiv:2510.04591 [eess.SY]},
  doi          = {10.48550/arXiv.2510.04591}
}

@misc{universal_de,
  author       = {Rackauckas, Christopher and Ma, Yingbo and Martensen, Julius and Warner, Collin and Zubov, Kirill and Supekar, Rohit and Skinner, Dominic and Ramadhan, Ali and Edelman, Alan},
  title        = {Universal differential equations for scientific machine learning},
  year         = {2020},
  howpublished = {arXiv preprint arXiv:2001.04385 [cs.LG]},
  doi          = {10.48550/arXiv.2001.04385}
}

@article{pavone_scaramuzza_pinn,
  author  = {Salzmann, Tim and Kaufmann, Elia and Arrizabalaga, Jon and Pavone, Marco and Scaramuzza, Davide and Ryll, Markus},
  title   = {Real-time neural {MPC}: deep learning model predictive control for quadrotors and agile robotic platforms},
  journal = {IEEE Robot. Autom. Lett.},
  year    = {2023},
  volume  = {8},
  number  = {4},
  pages   = {2397--2404},
  doi     = {10.1109/LRA.2023.3246839}
}

@inproceedings{tri_pinn,
  author    = {Ding, Nan and Thompson, Michael and Dallas, James and Goh, Jonathan Y. M. and Subosits, John K.},
  title     = {Drifting with unknown tires: learning vehicle models online with neural networks and model predictive control},
  booktitle = {2024 IEEE Intelligent Vehicles Symposium (IV)},
  year      = {2024},
  pages     = {2545--2552},
  doi       = {10.1109/IV55156.2024.10588474}
}

@misc{bosch_pinn,
  author       = {Rhode, Stephan and Jarmolowitz, Fabian and Berkel, Felix},
  title        = {Vehicle single track modeling using physics guided neural differential equations},
  year         = {2024},
  howpublished = {arXiv preprint arXiv:2403.11648 [cs.CE]},
  doi          = {10.48550/arXiv.2403.11648}
}

@article{pinn_domain_of_validity,
  author  = {Elsheikh, Mohamed and Ortmanns, Yak and Hecht, Felix and Ro{\ss}mann, Volker and Kr{\"a}mer, Stefan and Engell, Sebastian},
  title   = {Model predictive control of an industrial distillation column based on a hybrid model: adapting the domain of validity},
  journal = {IFAC-PapersOnLine},
  year    = {2023},
  volume  = {56},
  number  = {2},
  pages   = {7166--7171},
  note    = {22nd IFAC World Congress},
  doi     = {10.1016/j.ifacol.2023.10.597}
}

@inproceedings{fork_rl,
  author    = {Oishi, Koshi and Kato, Teruki and Makino, Hiroya and Ito, Seigo},
  title     = {Visual-based forklift learning system enabling zero-shot {Sim2Real} without real-world data},
  booktitle = {2025 IEEE International Conference on Robotics and Automation (ICRA)},
  year      = {2025},
  pages     = {4915--4921},
  doi       = {10.1109/ICRA55743.2025.11127682}
}

@article{samad2020industry,
  author  = {Samad, Tariq and Bauer, Margret and Bortoff, Scott and Di Cairano, Stefano and Fagiano, Lorenzo and Odgaard, Peter Fogh and Rhinehart, R. Russell and S{\'a}nchez-Pe{\~n}a, Ricardo and Serbezov, Atanas and Ankersen, Finn and Goupil, Philippe and Grosman, Benyamin and Heertjes, Marcel and Mareels, Iven and Sosseh, Raye},
  title   = {Industry engagement with control research: Perspective and messages},
  journal = {Annual Reviews in Control},
  year    = {2020},
  volume  = {49},
  pages   = {1--14},
  doi     = {10.1016/j.arcontrol.2020.03.002}
}

\end{document}